\documentclass[journal]{IEEEtran}
\IEEEoverridecommandlockouts
\usepackage{cite}
\usepackage{amsmath,amssymb,amsfonts}
\usepackage{algorithmic}
\usepackage{algorithm}
\usepackage{graphicx}
\usepackage{multirow}
\usepackage{subfigure}
\usepackage{textcomp}
\usepackage{xcolor}
\usepackage{amsthm}

\def\BibTeX{{\rm B\kern-.05em{\sc i\kern-.025em b}\kern-.08em
    T\kern-.1667em\lower.7ex\hbox{E}\kern-.125emX}}
\begin{document}

\title{A Genetic Algorithm based Superdirective Beamforming Method under Excitation \\Power Range Constraints \\
{\footnotesize}

}

\author{Jingcheng Xie,
Haifan Yin,~\IEEEmembership{Member,~IEEE,}
and Liangcheng Han
\thanks{Jingcheng Xie, Haifan Yin and Liangcheng Han are with Huazhong University of Science and Technology,
Wuhan 430074, China (e-mail: xiejc@hust.edu.cn; yin@hust.edu.cn; hanlc@hust.edu.cn). }
}
\maketitle

\begin{abstract}

The array gain of a superdirective antenna array can be proportional to the square of the number of antennas. However, the realization of the so-called superdirectivity entails accurate calculation and application of the excitations. Moreover, the excitations require a large dynamic power range, especially when the antenna spacing is smaller. In this paper, we derive the closed-form solution for the beamforming vector to achieve superdirectivity. We show that the solution only relies on the data of the array electric field, which is available in measurements or simulations. In order to alleviate the high requirement of the power range, we propose a genetic algorithm based approach with a certain excitation range constraint. Full-wave electromagnetic simulations show that compared with the traditional beamforming method, our proposed method achieves greater directivity and narrower beamwidth with the given range constraints.
\end{abstract}

\begin{IEEEkeywords}
superdirectivity, beamforming, excitation range constraint, spherical wave expansion, genetic algorithm
\end{IEEEkeywords}

\section{Introduction}
As one of the enabling technologies of the fifth generation mobile communication (5G), massive multiple-input multiple-output (MIMO)  is the key to boost the spectral efficiency with a large number of antennas at the base station \cite{Marzetta2010TWC,Yin2013JSAC}. However, the increase in the number of antennas often leads to the inevitable rise of the array size since the antenna spacing is generally no less than half a wavelength. In recent years, with the greater demand for the spectral efficiency, researchers begin to discuss the possibility of super dense antenna arrays \cite{Pizzo2020JSAC}\cite{Pizzo2022TWC}. In this scenario, the mutual coupling between array elements is no longer negligible yet helpful. Uzkov has proved in \cite{Uzkov} that the directivity of a linear array with $M$ isotropic antennas can reach $M^2$ as the spacing between antennas tends to zero. Hence, for base stations with many antennas, the array gain can be more significant \cite{Marzetta_superdirective}.\par
Despite the potential improvement of the array gain, the precise calculation of the superdirective beamforming vector is a challenging problem. The work \cite{Altshuler2005TAP} derives the beamforming vector and measures the directivity of a two-element array. However, the derivation ignores the field distortion caused by the strong mutual coupling. The author of \cite{Clemente2015TAP} designs a four-element parasitic superdirective array. The beamforming vector is calculated using the spherical wave expansion (SWE), which may lead to high calculation complexity as the number of antennas increases. A prototype of the superdirective antenna array based on impedance coupling and field coupling \cite{Han2022ICC} is built and measured in \cite{Han2023JSAC}, whose beamforming vector is corrected by the coupling matrix that still need to be calculated in the measurement. Furthermore, the work \cite{Altshuler2005TAP} shows that the required amplitude range of the beamforming vector increases as the antenna spacing decreases and the number of antennas grows. It shows another practical challenge that the wide range of the amplitude of the beamforming vector usually exceeds the linear range of the power amplifiers, which undermines the practical value of superdirectivity. To our best knowledge, this problem has not been addressed so far in the open literature.\par
In this paper, we first derive the beamforming vector starting from SWE and obtain a more concise closed-form solution which is only related to the electric field. Moreover, based on the derived solution, we propose an approach utilizing the idea of genetic algorithm (GA) in order to alleviate the problem of the wide power range requirement for the beamforming vector. 
 GA is widely used in antenna array optimization \cite{GA_intro,GA_superdirective2019}. However, this paper is the first to ultilize the idea to obtain the beamforming vector with excitation range constraints. Finally, the results are simulated under full-wave electromagetic simulations. The results show that, compared with the traditional method, our proposed method achieves greater directivity and narrower beamwidth with the given range constraints.
 
\section{Derivation of the superdirective beamformer}
In this section, we derive the beamforming vector for superdirective arrays under the framework of spherical wave expansion. The SWE is firstly introduced by Hansen \cite{SWE_first} to generate solutions to the vector wave equation. Then a detailed formulation and derivation is given by Stratton \cite{SWE_formulation}. It can decompose the electromagnetic field into a series of orthogonal spherical wave basis. In this method, the electric field can be expanded in the spherical coordinates $(r,\theta ,\phi )$ as \cite{SWE_derivation}
\begin{equation}
    \vec E (r,\theta ,\phi ) = \frac{k}{\sqrt \eta  }\sum_{s=1}^{2} \sum_{n=1}^{\infty }\sum_{m=-n}^{n} Q_{smn}\vec{F}_{smn}^{(3)}(r,\theta,\phi) , 
\end{equation}
where $\eta$ is the medium intrinsic impedance, $k$ is the wavenumber. $Q_{smn}$ is the spherical wave coefficient, and $\vec{F}_{smn}^{(3)}(r,\theta,\phi)$ is the wave function, where $s,m,n$ denote the wave modes.\par
In the far-field region, as $kr\to \infty$, the electric field $\Vec{E}$ can be simplified to \par
\begin{equation}
      \vec E (r,\theta ,\phi ) \to  \frac{k}{\sqrt \eta  }\frac{e^{ikr}}{\sqrt[]{4\pi}kr}\sum_{s=1}^{2}\sum_{n=1}^{\infty } \sum_{m=-n}^{n} Q_{smn}\vec{K}_{smn}(\theta,\phi),
\end{equation}
where $\vec{K}_{smn}(\theta,\phi)=\lim_{kr \to \infty}[\sqrt[]{4\pi}\frac{kr}{e^{ikr}}\vec{F}_{smn}^{(3)}(r,\theta,\phi)]$ are the far-field pattern functions. Their explicit expressions are
\begin{equation}
\begin{split}
     \vec{K}_{1mn}(\theta,\phi)&=\sqrt{\frac{2}{n(n+1)}}(-\frac{m}{\left | m \right|}) ^{m}e^{jm\phi}(-j)^{n+1} \\
    &\left \{ \frac{jm\bar{P}_{n}^{\left | m \right | }(\cos\theta)}{\sin\theta }\hat{\theta }-\frac{\mathrm{d}\bar{P}_{n}^{\left | m \right | }(\cos\theta)}{\mathrm{d}\theta }\hat{\phi }  \right \},
\end{split}
\end{equation}
\begin{equation}
\begin{split}
     \vec{K}_{2mn}(\theta,\phi)&=\sqrt{\frac{2}{n(n+1)}}(-\frac{m}{\left | m \right|}) ^{m}e^{jm\phi}(-j)^{n} \\
    &\left \{ \frac{\mathrm{d}\bar{P}_{n}^{\left | m \right | }(\cos\theta)}{\mathrm{d}\theta}\hat{\theta }+\frac{jm\bar{P}_{n}^{\left | m \right | }(\cos\theta)}{\sin\theta }\hat{\phi }  \right \},
\end{split}
\end{equation}
where $\bar{P}_{n}^{\left | m \right | }$ is the associated normalized Legendre function. Finally, the SWE of the electric field in the far-field region can be represented as \cite{SWE_matrix}
\begin{equation}
    \vec{E}(\theta,\phi)=k\sqrt[]{\eta }\sum_{s=1}^{2}\sum_{n=1}^{\infty }\sum_{m=-n}^{n}Q_{smn}\vec{K}_{smn}(\theta ,\phi )  . 
\end{equation}\par
For a certain antenna, the power radiated per unit solid angle in the given direction is defined as 
\begin{equation}
\begin{split}
    P_{A}(\theta,\phi)&=\frac{1}{2}r^{2}\eta\left | \vec{E}(r,\theta ,\phi )  \right |^{2}\\
    &=\frac{1}{2}\frac{1}{4\pi}\left | \sum_{s=1}^{2}\sum_{n=1}^{\infty }\sum_{m=-n}^{n}Q_{smn}\vec{K}_{smn}(\theta ,\phi ) \right |^{2}.
\end{split}
\end{equation}
As for an isotropical radiator, the power radiated per unit solid angle is equal to the total radiated power divided by $4\pi$, namely
\begin{equation}
    P_{i}=\frac{P_{total}}{4\pi}=\frac{1}{2}\frac{1}{4\pi}\left | \sum_{s=1}^{2}\sum_{n=1}^{\infty }\sum_{m=-n}^{n}Q_{smn} \right |^{2} .
\end{equation}
Therefore, the directivity of the antenna in the given direction $(\theta,\phi)$ is defined as
\begin{equation}
\begin{split}
    &\quad\quad\quad\quad\quad D(\theta,\phi)=\frac{P_{A}(\theta,\phi)}{P_{i}}\\
    &=\frac{\left | \sum_{s=1}^{2}\sum_{n=1}^{\infty }\sum_{m=-n}^{n}Q_{smn}\vec{K}_{smn}(\theta ,\phi ) \right |^{2} }{ \sum_{s=1}^{2}\sum_{n=1}^{\infty }\sum_{m=-n}^{n}\left |Q_{smn} \right |^{2}}   .
\end{split}
\end{equation}\par
The maximization of (8) can be obtain by applying the Cauchy-Schwartz inequality \cite{Clemente2015TAP}. However, the calculation of $\vec{K}_{smn}(\theta ,\phi )$ and $Q_{smn}$ is very complicated. Starting from this expression, we derive a more concise expression of directivity and solution to the beamforming vector. For an antenna array with $M$ elements, let $\mathbf{b}=[b_{1},b_{2},\cdots,b_{M}]^{\mathrm{T}}\in \mathbb{C}^{M\times 1}$ denotes the beamforming vector, where $b_{i},i=1,2,\cdots,M$ represents the excitation coefficent on the $i$-th antenna. $\mathbf{E}=[\mathbf{e} _{1},\mathbf{e} _{2},\cdots,\mathbf{e} _{M}]\in \mathbb{C}^{2ql\times M}$ represents the electric field in the quantified angle of each antenna, in which $q$ and $l$ denote the discrete angles in the $\theta$ and $\phi$ direction. $\mathbf{E}_{\theta_{0},\phi_{0}}=[{E}_{1}(\theta_{0},\phi_{0}), {E}_{2}(\theta_{0},\phi_{0}),\cdots,{E}_{M}(\theta_{0},\phi_{0})]^{\mathrm{T}}$ is the electric field of each antenna in the given direction $(\theta_{0},\phi_{0})$.
\newtheorem{theorem}{Theorem}
\begin{theorem}\label{them1}
    For the antenna array with $M$ elements, the superdirective beamforming vector that achieves the maximum directivity in the given direction $(\theta_{0},\phi_{0})$ can be obtained by eigenvector decomposition of the following matrix
    \begin{equation}
    (\mathbf{E}^{\mathrm{H}}\mathbf{E})^{-1}(\mathbf{E}_{\theta_{0},\phi_{0}}\mathbf{E}_{\theta_{0},\phi_{0}}^{\mathrm{H}})
\end{equation} 
The corresponding directivity is  
\begin{equation}
    D(\theta_{0},\phi_{0})=\frac{\mathbf{b}^{\mathrm{T}}\mathbf{E}_{\theta_{0},\phi_{0}}\mathbf{E}_{\theta_{0},\phi_{0}}^{\mathrm{H}}\mathbf{b}^{\ast}}{\mathbf{b}^{\mathrm{T}}\mathbf{E}^{\mathrm{H}}\mathbf{E}  \mathbf{b}^{\ast}}\cdot c,
    \end{equation}
    where $c$ is a constant. 
\end{theorem}
\begin{proof}
    The proof can be found in Appendix \ref{ap1}.
\end{proof}
Theorem \ref{them1} indicates that the solution to the superdirective beamforming vector only relies on the electric field of the array element. Such information can be obtained by simulations or experimental measurements in anechoic chambers.
\section{Superdirectivity with excitation range constraints}
The entries of the obtained superdirective beamforming vector $\mathbf{b}$ generally have a wide range of amplitude, especially when the number of antennas $M$ increases and the antenna spacing decreases. In this section, we propose a solution under a certain range constraint of the amplitude. The problem can be described as
\begin{equation}
\begin{aligned}  
&\max_{\mathbf{b}}\quad 
f(\mathbf{b})=\frac{\mathbf{b}^{\mathrm{T}}\mathbf{E}_{\theta_{0},\phi_{0}}\mathbf{E}_{\theta_{0},\phi_{0}}^{\mathrm{H}}\mathbf{b}^{\ast}}{\mathbf{b}^{\mathrm{T}}\mathbf{E}^{\mathrm{H}}\mathbf{E}  \mathbf{b}^{\ast}}\\
& \begin{array}{r@{\quad}r@{}l@{\quad}l}
s.t.&\quad\frac{\max(\left | b_{i} \right |)}{\min(\left | b_{i} \right |)} &\leq P  &i=1,2,3, \cdots,M,\\
\end{array} 
\end{aligned}
\end{equation}
where $P$ is the given range of amplitude. Without loss of generality, the minimum amplitude is normalized to 1, and the problem can be rewritten as 
\begin{equation}
\begin{aligned}
&\max_{\mathbf{b}}\quad f(\mathbf{b})=\frac{\mathbf{b}^{\mathrm{T}}\mathbf{E}_{\theta_{0},\phi_{0}}\mathbf{E}_{\theta_{0},\phi_{0}}^{\mathrm{H}}\mathbf{b}^{\ast}}{\mathbf{b}^{\mathrm{T}}\mathbf{E}^{\mathrm{H}}\mathbf{E}  \mathbf{b}^{\ast}}\\
& \begin{array}{r@{\quad}r@{}l@{\quad}l}
s.t.&\quad 1\leq \left | b_{i} \right |\leq P  \quad &i=1,2,3, \cdots, M.\\
\end{array} 
\end{aligned}
\end{equation}\par
The objective function and the constraint of the above optimization problem are non-convex and thus it is difficult to solve directly. Therefore, we proposed an approach to solve the above problem with the idea of Genetic Algorithm. The algorithm can be summarized as generating the initial population, calculating the fitness function, selecting the candidates and repoducing by crossover and mutation.\par
First, we randomly generate a set of $I$ beamforming vectors. For each vector $\mathbf{b}$, we choose $f(\mathbf{b})$ as the fitness function and calculate $f(\mathbf{b})$ to form a $I$-dimensional array. Then we sort the array from the largest to the smallest, selecting the beamforming vectors corresponding to the first $m$ of the sorted array. These $m$ vectors are remained for evolution. Noting that each vector has $M$ complex numbers containing the amplitude and phase, we separately encode the amplitude and phase of each complex number. Fig. \ref{fig2} shows the detailed encoding process. Both amplitude and phase are encoded as a binary sequence and then combined as a chromosome. We can control the unit quantization value of the encoding of the amplitude to ensure that the amplitude satisfies the constraint. For instance, in case the excitation range constraint is $P$ and the amplitude is encoded to $x$ bits, the unit quantization value should be $\frac{P-1}{2^x-1}$.\par
After the encoding stage, a population consisting of $m$ initial candidate beamforming vectors are generated and each vector consists of $M$ chromosomes. These candidates are used to reproduce ``children'' by crossover from two randomly selected candidates over each of their $M$ chromosomes. The reproduction procedure mainly includes two genetic operations: crossover and mutation. Parents are randomly picked in the candidate pool and mated. For each of the $M$ chromosomes, two random crossover point are selected. Then the chromosome fragment between the two points in the corresponding chromosome of one parent is swapped into that of the other to generate child chromosome. After the same operation over all $M$ chromosomes of the parent, a child solution is reproduced. For each reproduced chromosome, a mutation process may happen that converts a bit to the opposite one with a very low probability. Fig. \ref{fig3} shows the crossover and mutation process in which a chromosome has 10 bits of amplitude and 8 bits of phase. We repeat the above process until the population increase from $m$ to $I$. 
\begin{figure}[htbp]
\includegraphics[width=5.5cm]{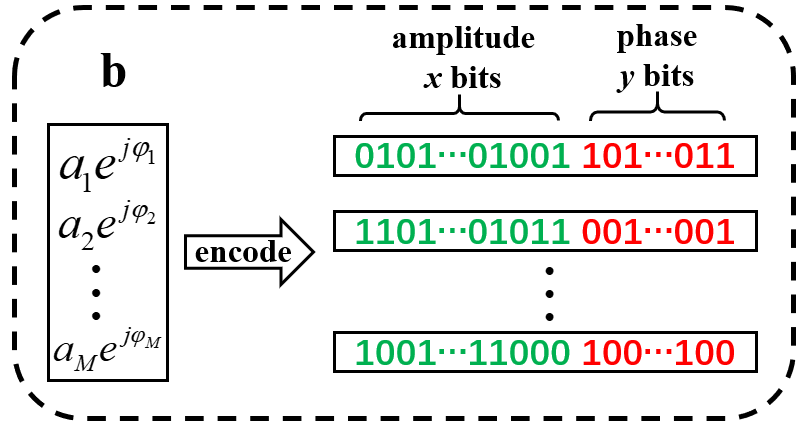}
\centering
\caption{The encoding process of a randomly selected vector $\mathbf{b}$.}
\label{fig2}
\end{figure}
\vspace{-0.5cm}
\begin{figure}[htbp]
\includegraphics[width=6cm]{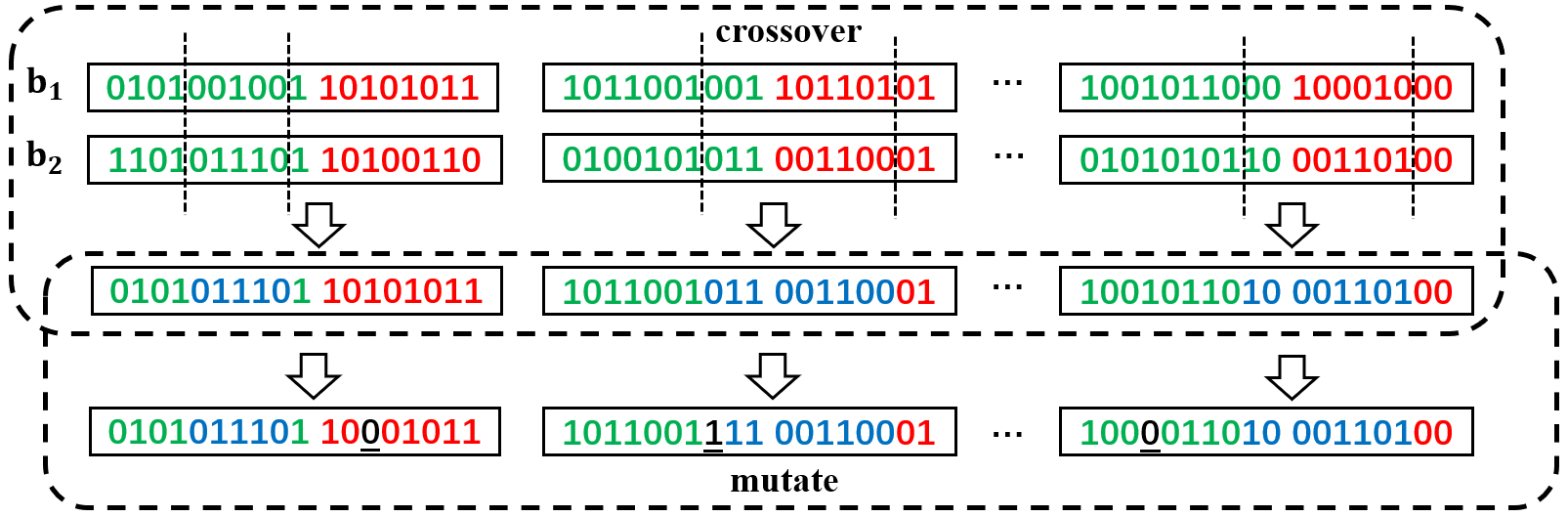}
\centering
\caption{The crossover and mutation of two parents with $M$ chromosomes.}
\label{fig3}
\end{figure}\par
For the reproduced set $I$, the selection and reproduction are repeated until the termination conditions are met. In general, the conditions can be either the biggest $f(\mathbf{b})$ is achieved or there is no further improvement in the successive iteration. Finally, a pseudo code of the proposed method is summarized in Algorithm 1.
\begin{algorithm}
    \renewcommand{\algorithmicrequire}{\textbf{Input:}}
	\renewcommand{\algorithmicensure}{\textbf{Output:}}
    \caption{Proposed GA based Beamforming Algorithm }
    \label{AL1}
    \begin{algorithmic}[1]
        \REQUIRE Excitation range constraint $P$, The randomly generated beamforming vector set $\mathfrak{B} =\left\{\mathbf{b}_{1},\mathbf{b}_{2},\cdots,\mathbf{b}_{I}\right\}$ which meets the constraint.\par
        $x$: Number of amplitude bits\par
        $y$: Number of phase bits\par
        $f(\mathbf{b})$: Fitness function\par
        \WHILE{the maximum number of iterations is not reached or the best fitness $f(\mathbf{b})\le f_{max}$ }
        \STATE (Select)
        \FOR{$\mathbf{b}_{i}\in \mathfrak{B}$}
        \STATE calculate $f(\mathbf{b}_{i})$;
        \ENDFOR
        \STATE select $m$ candidates;
        \STATE (Coding by the constraint)
        \FOR{$\mathbf{b}_{i}\in$ candidates}
        \FOR{$b_{j}^{i}\in \mathbf{b}_{i}$}
        \STATE Encode $b_{j}^{i}$ into a chromosome in which $x$ bits denote amplitude and $y$ bits denote phase. 
        \ENDFOR
        \ENDFOR
        \STATE (Crossover and Mutation)
        \IF{Number of candidates $<I$}
        \STATE Randomly selected two candidates and perform crossover and mutation for all $M$ chromosomes.
        \ENDIF
        \STATE (Decoding)
        \STATE Inverse the coding process;
        \ENDWHILE
        \STATE Select the best candidate;
        \ENSURE The beamforming vector $\mathbf{b}^{*}$.
    \end{algorithmic}
\end{algorithm}
\section{Numerical Results}
To prove the effectiveness of the proposed beamforming algorithm, full-wave simulations are carried out in this section.\par
The simulation is performed at 1.6 $\mathrm{GHz}$ considering two different arrays (four and six identical uniformly spaced electrical dipoles). The designed dipole antenna array is shown in Fig. \ref{fig4}. The array is printed on Rogers RO4003C (lossy) substrate ($\varepsilon_{r}=3.55$, $\mu_{r}=1$, $\tan \delta=0.0027$, the width $L=85.5$ mm and the thickness is $0.813$ mm). The length and width of the dipole antenna are $H=71.48$ mm and $w=1$ mm respectively. The caliber $h$ of the port is $2.54$ mm for connection to a SubMiniature version A (SMA) connector. The distance between two adjacent antennas is set to $0.1\lambda$ or $0.2\lambda$ and the designed end-fire direction is $(\theta_{0}=90^{\circ },\phi_{0}=0)$.\par
Considering the mutual coupling effects between each array element, the electric field $\mathbf{e}_{i}=\begin{bmatrix}
     \bar{E}(\theta_{1},\phi_{1})_{\theta}& \cdots &\bar{E}(\theta_{l},\phi_{q})_{\theta}
    \end{bmatrix}^{T}\in \mathbb{C}^{lq\times 1}$ in (17) of each antenna is simulated, from which the electric field $\mathrm{E}_{\theta_{0},\phi_{0}}$ in the end-fire direction $(\theta_{0}=90^{\circ },\phi_{0}=0)$ is extracted. Then, the simulated complex electric fields is used to calculate the fitness function $f(\mathbf{b})$ (see section \uppercase\expandafter{\romannumeral3}-A). Finally, we apply our proposed algorithm to obtain the beamforming vectors $\mathbf{b}$,  based on which the radiation pattern is simulated. 
 In the simulation, we let a 7-bit binary code denote the amplitude and 0.01, 0.02, 0.03 per unit, respectively. Therefore, we choose $P=2.27$ for 4 antennas with $0.1\lambda$ spacing and $P=2.27$, $3.54$, $4.81$ for 6 antennas with $0.2\lambda$ spacing. To show that our proposed method is efficient with all constraint $P$, the maximum ratio transmission (MRT) and the traditional superdirecitve beamforming that ignores the field distortion due to mutual coupling
 are chosen for comparison.\par

 \begin{figure}[htbp]
\includegraphics[width=6cm]{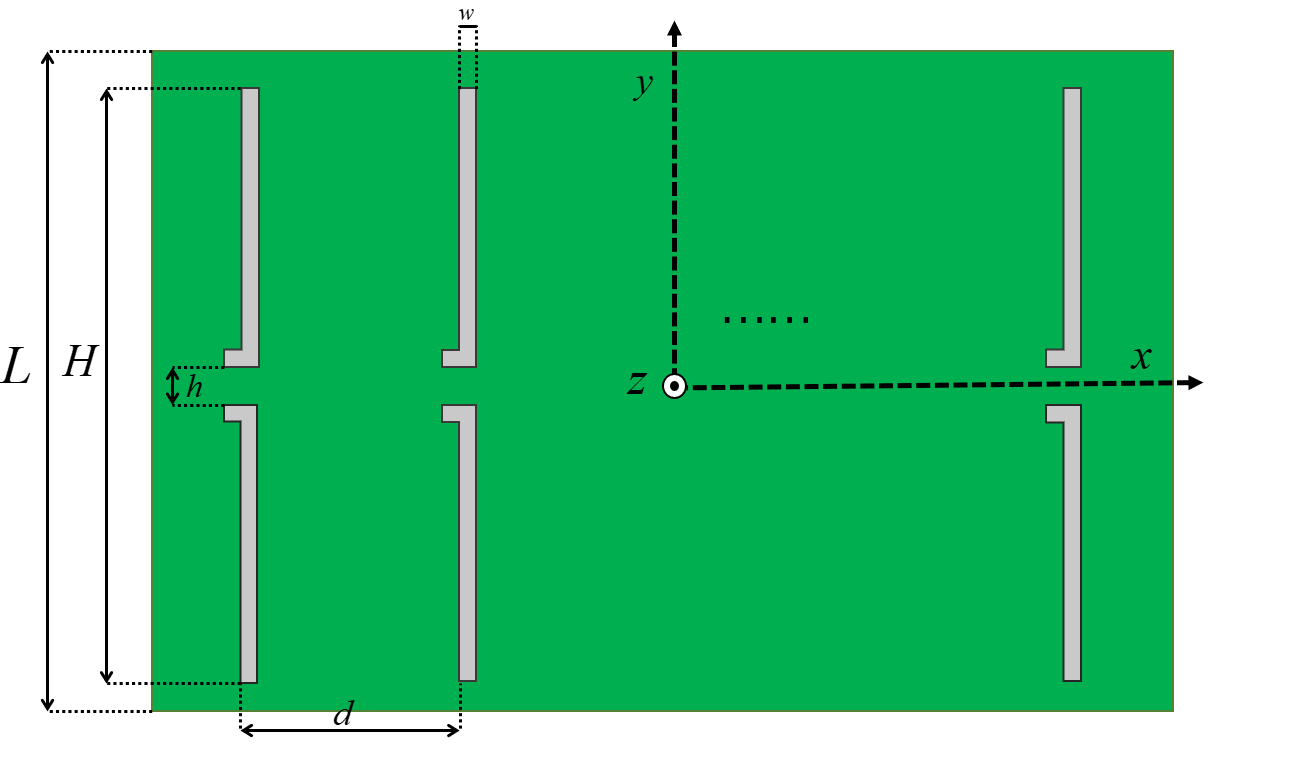}
\centering
\caption{The schematic view of the designed dipole antenna array.}
\label{fig4}
\end{figure}
The simulated result in the E-plane ($\phi=0^{\circ }$) of the 4 antennas is shown in Fig. \ref{fig8}.
The directivity simulated using eigenvalue decomposition in the end-fire direction is 16.45 and the beamwidth is 51$^{\circ}$.
 With the constraint $P=2.27$, our method achieves the directivity of 11.33, which is much higher than 4.45 of the MRT and 6.15 of the traditional method. Moreover, it can be found that the 3-dB beamwidth of our proposed method is 62.7$^{\circ}$, which is narrower than 132.6$^{\circ}$ of the MRT method and 68.4$^{\circ}$ of the traditional method.
\begin{figure}[htbp]
\centering
\includegraphics[width=6.5cm]{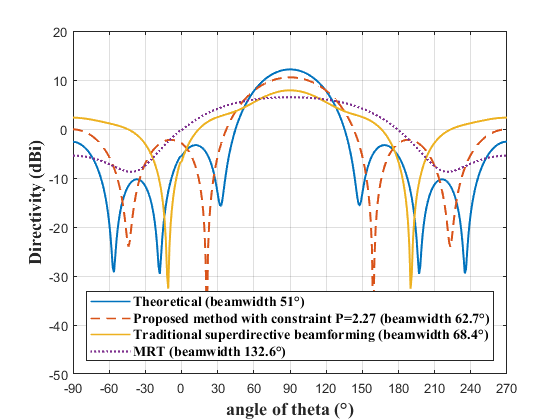}
\caption{The E-plane pattern of the $4$-element array with $0.1\lambda$ spacing.}
\label{fig8}
\end{figure}

Then, we increase the number of antennas to 6 and the spacing is changed to $0.2\lambda$. The simulated directivity pattern in the E-plane ($\phi=0^{\circ }$) is shown in Fig. \ref{fig9}. It is obvious that our proposed method with all given constraints has a better performance in the main-lobe direciton than the traditional method and MRT.
\begin{table}[htbp]
    \centering
    \caption{The simulated results of the 6-element array}
    \label{tab2}
    \resizebox{\linewidth}{!}{
    \begin{tabular}{c|cccccc}
    \hline
     $I=6$ & \multirow{2}{*}{Theoretical} & \multirow{2}{*}{$P=2.27$} & \multirow{2}{*}{$P=3.54$} & \multirow{2}{*}{$P=4.81$} &
     Traditional &
     \multirow{2}{*}{MRT}\\
     $d=0.2\lambda$ &  &  &  &  & method & \\
    \hline
    Directivity & 31.48 & 23.9 & 24.8 & 25.5 & 10.2 & 8.5\\
    \hline
    3-dB beamwidth ($^{\circ}$) & 37.3 & 42.2 & 41.3 & 40.9 & 58.4 & 82.1\\
    \hline
    \end{tabular}
    }
\end{table}
\begin{table}[htbp]
    \centering
    \caption{The simulated results of the 8-element array}
    \label{tab3}
    \resizebox{\linewidth}{!}{
    \begin{tabular}{c|cccccc}
    \hline
     $I=8$ &\multirow{2}{*}{Theoretical}& \multirow{2}{*}{$P=2.27$} & \multirow{2}{*}{$P=3.54$} & \multirow{2}{*}{$P=4.81$} &
     Traditional               & 
     \multirow{2}{*}{MRT}\\
     $d=0.2\lambda$ &  &  &  &  & method &\\
    \hline
    Directivity & 57.23 & 30.05 & 30.09 & 32.00 & 4.02 & 10.71\\
    \hline
    3-dB beamwidth ($^{\circ}$) & 32 & 38.2 & 39.6 & 37.3 & 51.2 &71\\
    \hline
    \end{tabular}
    }
\end{table}
The detailed results are shown in Table \ref{tab2}, where the directivity is obtained in the end-fire direction $(\theta_{0}=90^{\circ },\phi_{0}=0)$. It can be found that our proposed method achieves greater directivity and much narrower beamwidth.

\begin{figure}[htbp]
\centering
\includegraphics[width=6.5cm]{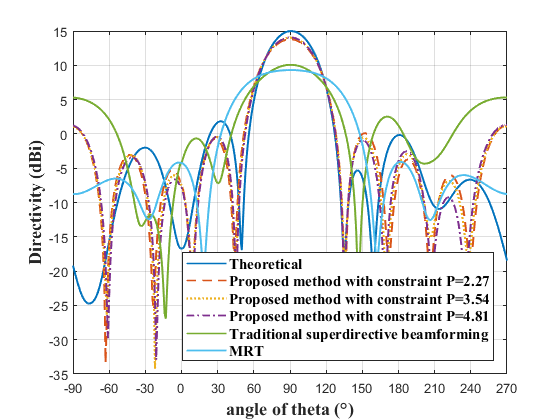}
\caption{The E-plane pattern of the $6$-element array with  $0.2\lambda$ spacing.}
\label{fig9}
\end{figure}

\begin{figure}[htbp]
\centering
\includegraphics[width=6.5cm]{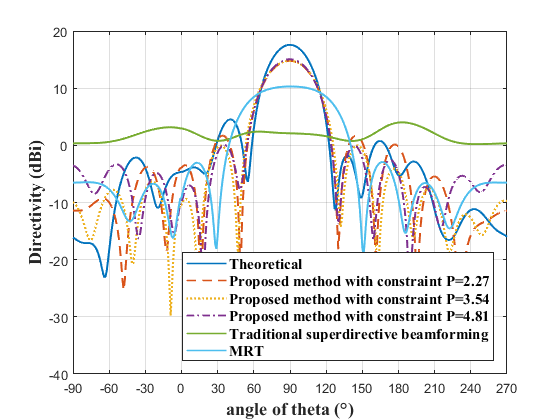}
\caption{The E-plane pattern of the $8$-element array with $0.2\lambda$ spacing.}
\label{fig10}
\end{figure}

In Fig. \ref{fig10}, the number of antennas is increased to 8 while the spacing maintains $0.2\lambda$.  The detailed results are listed in Table \ref{tab3}. It can be found that our proposed method is effective and feasible even the number of antenna increases.

\section{Conclusion}
In this paper, we derived the beamforming vector to achieve the superdirectivity, which can be calculated entirely from the electric field of the antenna array. Moreover, to alleviate the requirement requirement of the wide amplitude range for the beamforming vector, a GA-based effective algorithm is proposed to obtain a beamforming vector with a certain excitation range constraint. The simulated results showed that compared with the traditional superdirective beamforming method and the MRT, our proposed method achieves greater directivity.

\begin{appendices}
    \section{Proof of Theorem \ref{them1}}
    \label{ap1}

        For an antenna array with $M$ elements, the wave coefficient of each wave mode is the sum of that of each antenna element, which means $Q_{smn}=\sum_{i=1}^{M}b_{i}Q_{smni}$. Then the directivity of an $M$-element antenna array is 
        \begin{equation}
             D(\theta,\phi)\!=\!\frac{\left |\sum_{i=1}^{M}\!b_{i}\! \sum_{s=1\!}^{2}\sum_{n=1\!}^{\infty }\sum_{m=-n\!}^{n}\!Q_{smni}\vec{K}_{smn}\!(\theta ,\phi ) \!\right |^{2} }{ \sum_{s=1}^{2}\sum_{n=1}^{\infty }\sum_{m=-n}^{n}\left |\sum_{i=1}^{M}b_{i}Q_{smni} \right |^{2}}  .
        \end{equation}\par
        In order to simplify this expression, we transform it to the form of matrix. Let $\overline{Q}_{smni}=\begin{bmatrix}
        Q_{1,-1,1,i}&... &Q_{2,N,N,i}
        \end{bmatrix}\in \mathbb{C} ^{1\times T }  $ represent all modes of the wave coefficients of the $i$-th element, where $T=2\times N\times (N+2)$ and $N$ is the trunction point of $n$ since it has an infinite number of values \cite{SWE_derivation} \cite{Harrington1958TAP} \cite{Thal1988TAP}. Then the wave coefficient of every mode and every array element can be expressed as
        \begin{equation}
        \mathbf{Q}=\begin{bmatrix}
        Q_{1,-1,1,1} & ... &Q_{2,N,N,1} \\
        Q_{1,-1,1,2} & ... &Q_{2,N,N,2} \\
        ...& ... & ...\\
        Q_{1,-1,1,M} & ... &Q_{2,N,N,M}
        \end{bmatrix}\in \mathbb{C} ^{M\times T }.
        \end{equation}\par
        Similarly, we can also transform the far-field pattern functions into the form of the matrix
        \begin{equation}
        \mathbf{K}\!=\!\begin{bmatrix}
         K_{1,-1,1}(\theta_{1},\phi_{1})_{\theta } & ... &K_{2,N,N}(\theta_{1},\phi_{1})_{\theta } \\
         K_{1,-1,1}(\theta_{1},\phi_{1})_{\phi } & ... &K_{2,N,N}(\theta_{1},\phi_{1})_{\phi  } \\
         K_{1,-1,1}(\theta_{2},\phi_{1})_{\theta } & ... &K_{2,N,N}(\theta_{2},\phi_{1})_{\theta } \\
         K_{1,-1,1}(\theta_{2},\phi_{1})_{\phi } & ... &K_{2,N,N}(\theta_{2},\phi_{1})_{\phi } \\
         ...& ... & ...\\
         K_{1,-1,1}(\theta_{l},\phi_{q})_{\phi } & ... &K_{2,N,N}(\theta_{l},\phi_{q})_{\phi }
        \end{bmatrix}\!\in \mathbb{C} ^{ (2lq) \times T }
        \end{equation}
        where each row indicates the $\theta$ or $\phi$ component values of the far-field pattern functions of all modes at a given solid angle and $lq$ is the the number of the angle quantization points.\par
        Then the directivity expression (13) in a given direction $(\theta_{0},\phi_{0})$ can be rewritten as
        \begin{equation}
        D(\theta_{0},\phi_{0})=\frac{\mathbf{b}^{\mathrm{T}}\left [ \mathbf{Q}\mathbf{K}_{\theta_{0},\phi_{0}}^{\mathrm{T}} \right ]\left [ \mathbf{Q}\mathbf{K}_{\theta_{0},\phi_{0}}^{\mathrm{T}} \right ]^{\mathrm{H}}\mathbf{b}^{\ast} }{\mathbf{b}^{\mathrm{T}}\mathbf{Q}\mathbf{Q}^{\mathrm{H}}\mathbf{b}^{\ast}} ,
        \end{equation}
        where $\mathbf{b}=[b_{1},b_{2},\cdots,b_{M}]^{\mathrm{T}}\in \mathbb{C}^{M\times 1}$ is the beamforming vector and $\mathbf{K}_{\theta_{0},\phi_{0}}$ is the far-field pattern function in the given direction $(\theta_{0},\phi_{0})$.\par
        Moreover, if we quantify the electric field of the $i$-th antenna at the same points
        \begin{equation}
            \mathbf{e}_{i}\!=\!\begin{bmatrix}
             \bar{E}(\theta_{1},\phi_{1})_{\theta}& \bar{E}(\theta_{1},\phi_{1})_{\phi}& \cdots &\bar{E}(\theta_{l},\phi_{q})_{\phi}
            \end{bmatrix}^{T}\!\in \!\mathbb{C}^{2lq\times 1}.
        \end{equation}
        Then the SWE of the electric field in (5) becomes \cite{SWE_matrix}
        \begin{equation}
            \mathbf{e} _{i}=k\sqrt[]{\eta }\mathbf{K}\mathbf{Q}_{i}^{\mathrm{T}} ,
        \end{equation}
        where the subscript $i$ indicates the $i$-th antenna. By inserting (18) into (16), we obtain
        \begin{equation}
            D(\theta_{0},\phi_{0})=\frac{\mathbf{b}^{\mathrm{T}}\mathbf{E}_{\theta_{0},\phi_{0}}\mathbf{E}_{\theta_{0},\phi_{0}}^{\mathrm{H}}\mathbf{b}^{\ast}}{\mathbf{b}^{\mathrm{T}}\mathbf{Q}\mathbf{Q}^{\mathrm{H}}\mathbf{b}^{\ast}}\cdot \frac{1}{(k\sqrt[]{\eta } )^{2}} ,  
        \end{equation}
        where $\mathbf{E}_{\theta_{0},\phi_{0}}=[{E}_{1}(\theta_{0},\phi_{0}), {E}_{2}(\theta_{0},\phi_{0}),\cdots,{E}_{M}(\theta_{0},\phi_{0})]^{\mathrm{T}}$ is the electric field of each antenna in the given direction $(\theta_{0},\phi_{0})$.\par 
        Since the normalized far-field pattern functions are orthogonal \cite{SWE_formulation}, which means
        \begin{equation}
        \begin{split}
            \int_{0}^{2\pi }\int_{0}^{\pi} \vec{K}_{smn}(\theta,\phi)\vec{K}_{s'm'n'}(\theta,\phi)^{\ast}\mathrm{d}\theta\mathrm{d}\phi \\
            =\begin{cases}
          1& \text{ if } s=s',m=m',n=n' \\
          0& \text{ o.w } 
        \end{cases}
        \end{split}
        \end{equation}
By inserting (20), the expression (19) can be rewritten as 
        \begin{equation}
            D(\theta_{0},\phi_{0})=\frac{\mathbf{b}^{\mathrm{T}}\mathbf{E}_{\theta_{0},\phi_{0}}\mathbf{E}_{\theta_{0},\phi_{0}}^{\mathrm{H}}\mathbf{b}^{\ast}}{\mathbf{b}^{\mathrm{T}}\mathbf{E}^{\mathrm{H}}\mathbf{E}  \mathbf{b}^{\ast}}\cdot c,
        \end{equation}
        where $c$ is a constant related to the unit of the electric field and $\mathbf{E}=[\mathbf{e}_{1},\mathbf{e}_{2},...,\mathbf{e}_{M}]\in \mathbb{C}^{2ql\times M}$ is the electric field in the quantified angle of each antenna.\par
        The expression (21) has the form of a generalized Rayleigh quotient in the $M$-dimensional complex space $\mathbb{C}^{M}$. The maximization of the directivity can be obtained by the eigenvalue decomposition of the corresponding matrix \cite{Rayleigh_quotient}.
        \begin{equation}
            (\mathbf{E}^{\mathrm{H}}\mathbf{E})^{-1}(\mathbf{E}_{\theta_{0},\phi_{0}}\mathbf{E}_{\theta_{0},\phi_{0}}^{\mathrm{H}})\mathbf{x}=\lambda \mathbf{x}.
        \end{equation}
        Since $\mathbf{E}_{\theta_{0},\phi_{0}}$ represents the electric field of each array element in the given direction, which means the rank of the matrix $(\mathbf{E}_{\theta_{0},\phi_{0}}\mathbf{E}_{\theta_{0},\phi_{0}}^{\mathrm{H}})$ is 1, the above eigenvalue problem only has one solution. Let the eigenvalue of the matrix $(\mathbf{E}^{\mathrm{H}}\mathbf{E})^{-1}(\mathbf{E}_{\theta_{0},\phi_{0}}\mathbf{E}_{\theta_{0},\phi_{0}}^{\mathrm{H}})$ be $\lambda_{0}$ and the corresponding eigenvector be $\mathbf{x}_{0}$, then the maximum directivity is $D(\theta_{0},\phi_{0})_{\mathrm{max}}= \lambda_{0}\cdot c $ with beamforming vetor $\mathbf{b}=\mathbf{x}_{0}^{\ast }$. Finally, Theorem \ref{AL1} is proved.\qed
\end{appendices}

\bibliographystyle{IEEEtran}
\bibliography{ref1}
\end{document}